\documentclass[review]{elsarticle}

\usepackage{lineno,hyperref}
\modulolinenumbers[5]

\journal{Journal of \LaTeX\ Templates}
\usepackage{amssymb, amsfonts,amsmath, amsthm}
\usepackage{graphicx}
\usepackage{xcolor}
\usepackage{upgreek}
\usepackage{dsfont}
\usepackage{verbatim}
\usepackage{enumitem}
\usepackage{subcaption}

\newtheorem{theorem}{Theorem}[section]
\newtheorem{lemma}[theorem]{Lemma}

\begin{document}

\begin{frontmatter}

\title{Prevalence Estimation and Optimal Classification Methods to Account for Time Dependence in Antibody Levels}



\author[mymainaddress,mysecondaryaddress]{Prajakta Bedekar\corref{mycorrespondingauthor}}
\cortext[mycorrespondingauthor]{Corresponding author}
\ead{pbedeka1@jh.edu}
\author[mymainaddress]{Anthony J Kearsley}

\author[mymainaddress]{Paul N Patrone}

\address[mymainaddress]{Applied and Computational Mathematics Division, National Institute of Standards and Technology, 100 Bureau Drive, Gaithersburg, MD 20899, USA}
\address[mysecondaryaddress]{Department of Applied Mathematics and Statistics, Johns Hopkins University, 3400 N. Charles Street, Baltimore, MD 21218, USA}

\begin{abstract}
Serology testing can identify past infection by quantifying the immune response of an infected individual providing important public health guidance. Individual immune responses are time-dependent, which is reflected in antibody measurements. Moreover, the probability of obtaining a particular measurement changes due to prevalence as the disease progresses. Taking into account these personal and population-level effects, we develop a mathematical model that suggests a natural adaptive scheme for estimating prevalence as a function of time. We then combine the estimated prevalence with optimal decision theory to develop a time-dependent probabilistic classification scheme that minimizes error. We validate this analysis by using a combination of real-world and synthetic SARS-CoV-2 data and discuss the type of longitudinal studies needed to execute this scheme in real-world settings.
\end{abstract}

\begin{keyword}
Time-Dependent Classification \sep SARS-CoV-2 \sep Prevalence Estimation \sep Optimization
\sep  Antibody Testing
\end{keyword}

\end{frontmatter}



\section{Introduction}

Antibody tests are a powerful tool in the study of virus propagation throughout a population. They are useful for prevalence estimation (Bendavid et al., 2021 \cite{bendavid2021covid}, Burgess et al., 2020 \cite{burgess2020we}, Jacobsen et al., 2010 \cite{jacobsen2010hepatitis}), thus guiding public health decisions (Caini et al., 2020 \cite{caini2020meta}, Peeling et al., 2020\ \cite{peeling2020serology}). Antibody tests can also characterize immune response of an infected or vaccinated individual. Importantly, a classification scheme is needed to interpret raw measurement data and thereby determine whether a sample is deemed positive or negative.

A fundamental problem not addressed by traditional classification schemes is the multitude of timescales inherent to antibody kinetics. For example, the antibody levels in an individual vary on a `personal timeline'; after an exposure to a virus, titers peak following an initial delay and then slowly decay (Jin et al., 2020 \cite{jin2020diagnostic}, Sethuraman et al., 2020\ \cite{sethuraman2020interpreting}, Zhao et al., 2020 \cite{zhao2020antibody}). This behavior is person specific (Aydillo et al., 2021 \cite{aydillo2021immunological}, Dispinseri et al., 2021 \cite{dispinseri2021neutralizing}, Zhang et al., 2020 \cite{zhang2020viral}). Moreover, the number of infected people in the population changes with time (Dong et al., 2020 \cite{dong2020interactive}), which occurs on an `absolute timeline'. Traditional methods do not take these effects and their interactions into account, potentially misclassifying many samples due to a static training dataset and subsequent classification boundary. This can lead to inaccurate estimation of prevalence.

To address this problem, we develop a modeling approach that explicitly accounts for the relationship between time-dependent changes in antibody levels due to both individual- and population-level effects. This is done by recognizing that change in prevalence of the disease dictates what fraction of the population start personal timelines on a given day. As these models depend on prevalence, we develop an unbiased estimator for this quantity as a function of time. With this estimate, we construct an optimal classification scheme that minimizes the prevalence weighted average of false positive and false negative errors defined over time. We demonstrate the effectiveness of this adaptive scheme using a publicly available dataset representing SARS-CoV-2 antibody measurements.

A key challenge in developing time-dependent probability models is that the effects of personal timeline are convolved with those of the absolute timeline. Individual response probabilities are independent of the progress of a pandemic in the population, but the conditional probability of a sample on a given day being positive depends on the prevalence in the population. Prevalence is not known a priori, but this close relation to the conditional probabilities is leveraged to develop an unbiased estimator as a function of time. Interestingly, this estimator does not rely on classification, thereby allowing us to deconvolve prevalence from any test data.

One counterintuitive takeaway from our work is that the classification of samples depends on when it was collected, and any two samples from two different points in time could be classified differently, even if they have the same antibody measurements. Previous works (Patrone and Kearsley, 2021 \cite{patrone2021classification}) foreshadowed that changing prevalence leads to a variation in classification domains. However we also show that even with fixed prevalence, optimal classification domains can vary. 

A limitation of our approach is that we have primarily considered time to be discrete, as epidemiological data is often reported once per day. A continuous analog of these models is possible and discussed in detail in subsection \ref{subsec:cts}. Because of reporting constraints, delays in testing, and measurement error in the date of symptom onset, data sometimes has jumps. Changing the size of time discretization can help address these issues. See subsection \ref{subsec:optsynth} for further discussion.

\section{Methods: Theory of Time Dependent Classification}
\label{sec:theory}

Serology testing measures the quantities of $n$ antibody targets in a sample obtained from an individual, with results reported as fluorescence measurements. We represent such a measurement as a vector, $\textbf{r} = (r_1,r_2,\cdots,r_n)$, where the variables $r_i$ denote the values of antibody targets $i$, $i \in \{1,\cdots,n\}$.  These measurements fall in a domain $\Omega \subset \mathbb{R}^n$. Due to the limits of detection for the instruments, $\Omega$ is typically of the form $[l_1,u_1]\times \cdots, \times [l_n,u_n]$, where $l_i$ and $u_i$ are lower and upper limits of detection for target $i$.

\subsection{Probability Models}
\label{subsec:probmodel}
A sample value $\textbf{r}$ can arise from either a true positive or negative sample on a given day $T$ of a pandemic. We use $N(\mathbf{r})$ to denote the conditional probability density that a sample yields measurement $\mathbf{r}$ given that it is negative. Similarly, $P(\mathbf{r},T)$ is the conditional probability density that a sample yields measurement $\mathbf{r}$ on day $T$ of a pandemic given that it is positive. Given the aforementioned sources of time dependence, we assume
\begin{equation}
P: \Omega  \times \Omega_T \rightarrow \mathbb{R} \text{ and } N: \Omega \rightarrow \mathbb{R}.
\end{equation}

Ideally, $T=0$ is the day of the start of a pandemic, i.e.\ the first patient's exposure. It is extremely difficult to quantify this. As a result, in practice the personal and absolute timelines are both expressed in terms of the days since symptom onset. This issue is further considered in subsection \ref{subsec:modelingchoices}.

The conditional probability density for positive samples depends on time whereas for negative samples it does not. The following example helps clarify why. Consider an $\mathbf{r}$ close to $(u_1,u_2,\cdots,u_n)$ at $T=0$, i.e.\ a high antibody measurement before a pandemic has begun. The conditional probability $P$ of someone being infected is zero as there is no virus. It is likely that the value is merely due to nonspecific binding to antibodies for other viruses with a similar structure. On the other hand, if some people in the population are infected, this measurement could have come from an infected individual and thus should be classified as positive.

The positive samples on day $T$ of a pandemic come from those individuals who have been infected on some day before $T$. (For now, we set aside the question of antibody levels due to vaccine-induced immunity.) This sample space is thus partitioned into subsets based on how many days ago the individual was infected, i.e, the set of individuals who are on day $t$ after their day of infection, with $t = 0, 1, \cdots, T$. The law of total probability implies 
\begin{equation}
P(\mathbf{r},T) = \sum\limits_{t=0}^T Prob(\mathbf{r},T, \text{day t of infection}),
\end{equation}
where $Prob(\mathbf{r},T, \text{day t of infection})$ is the probability that a measurement $\mathbf{r}$ is obtained on day $T$ of a pandemic for an individual who was infected $t$ days ago. By using the multiplication rule (Ross,2010 \cite{ross2010first}) we see that the summand is the product of the probability that the antibody response for day $t$ of the infection is $\mathbf{r}$, and the fraction of population on day $t$ of their infection out of all positive individuals. Denoting the antibody response by $R$ and fraction of true positive individuals infected $t$ days ago as $f(T-t)$, one finds, 

\begin{equation}\label{eq:PrT}
P(\textbf{r},T) = \sum\limits_{t=0}^{T} R(\textbf{r},t) \ \frac{f(T-t)}{q(T)},
\end{equation}
where $R(\textbf{r},t)$ is the conditional probability of a sample having antibody measurement $\textbf{r}$ given that it is a positive sample, with the individual being infected $t$ days ago. The function $f(t)$ denotes the fraction of the population infected on day $t$ and $q(T) =  \sum\limits_{t=0}^{T} f(t)$ is the prevalence on day $T$ of a pandemic's timeline.

In practice, $R$ and $N$ are modeled from training data. It is important to choose a suitable family of probability distributions so that it closely models the disease progression. Guided by the limiting behavior of the antibody kinetics, the distributions for $R(.,0)$ and $\lim\limits_{t\rightarrow \infty} R(.,t)$ should be the same as that of $N(.)$ (Borremans, 2020 \cite{borremans2020quantifying}). Once such a form is determined, methods such as maximum likelihood estimates can be used to find the associated parameters. Section \ref{sec:COVID} illustrates this analysis in practice.

The probability $Q(\mathbf{r},T)$ that a sample collected on day $T$ has antibody level $\mathbf{r}$ can be found using $P(\mathbf{r},T), N(\mathbf{r}),$ and $q(T)$ via law of total probability,
\begin{equation}
Q(\mathbf{r},T) = q(T)\cdot P(\mathbf{r},T) + (1-q(T)) \cdot N(\mathbf{r}).
\end{equation}
Using change of variables and the fact that distributions of $N(\cdot)$ and $R(\cdot,0)$ are equal, we note for later convenience,
\begin{align} \label{eq:totprob}
Q(\mathbf{r},T) & = q(T)\cdot \sum\limits_{t=0}^T R(\mathbf{r},t) \frac{f(T-t)}{q(T)}\ + (1-q(T)) \cdot N(\mathbf{r})\\
			 & = N(\mathbf{r}) + \sum\limits_{t=0}^{T-1}  f(t) \cdot \left(R(\mathbf{r},T-t) - N(\mathbf{r})\right).
\end{align}

It is important to note that the domain for time, $\Omega_T$ can be taken to be continuous, i.e.\ $[0,T]$ and a corresponding model is outlined in subsection \ref{subsec:cts}. However, it is not possible to collect data for such a domain. In practice, data is collected and reported every day, leading to our choice of $\Omega_T = \{ 0, 1, \cdots, T\}$, a discrete set of days. Sometimes data is not available even at this granularity and coarser periods like weeks can be used. We discuss the issues associated with the discretization choice in subsection \ref{subsec:optsynth}.

\subsection{Estimation of Time-Dependent Prevalence}
\label{subsec:qest}

We develop a scheme for estimation of time-dependent prevalence based solely on the sample values measured over time. Provided that we have already modeled the probability distributions, this scheme does not require that samples be classified, i.e.\ prevalence can (and should) be estimated before constructing optimal classification domains. 

For an arbitrary partition of the domain $\Omega$ into $D_P$ and $D_N$, we define $Q_P$ at time $T$ by integrating both sides of total probability $Q$ from equation \eqref{eq:totprob} over $D_P$. This yields the total probability mass of $Q$ inside $D_P$,
\begin{equation}\label{eq:QP} Q_P(T) = \int_{D_P} Q(\mathbf{r},T) d\mathbf{r} = N_P + \sum\limits_{t=0}^{T-1}  f(t) \left(R_P(T-t) - N_P\right),
\end{equation}
where $R_P(T-t) = \int_{D_P} R(\mathbf{r},T-t) d\mathbf{r}$ and $N_P= \int_{D_P} N(\mathbf{r}) d\mathbf{r}$. Note that by assumption, $R_P$ and $N_P$ are calculated exactly. Furthermore, as the distribution for $R(\cdot,0)$ and $N(\cdot)$  are assumed to be identical, (subsection \ref{subsec:probmodel}), there is no term corresponding to $t=T$ in this sum.

Here, $Q_P(T)$ is approximated by a Monte-Carlo estimate
\begin{equation}\label{eq:QPestimate}
Q_P(T) \approx \widehat{Q}_P(T) = \frac{1}{S} \sum\limits_{i=1}^S \mathds{1}_P(\mathbf{r_i}),
\end{equation}
where $\mathds{1}_P$ is the indicator function of $D_P$ and $\{ \mathbf{r_1}, \mathbf{r_2}, \cdots, \mathbf{r_S} \}$ is the set of sample values observed on day $T$. Notice that $\widehat{Q}_P(T)$  is an unbiased estimator of $Q_P(T)$ (Caflisch,1998 \cite{caflisch1998monte}). 

Using $Q_P(1)$ and equation \eqref{eq:QP},
\begin{equation}
f(0) = \frac{Q_P(1)-N_P}{R_P(1)-N_P},
\end{equation}
we obtain the following estimator,
\begin{equation} \label{eq:f0est}
f(0) \approx \widehat{f}(0) = \frac{\widehat{Q}_P(1)-N_P}{R_P(1)-N_P}.
\end{equation}

Observe that due to the realistic assumption made in subsection \ref{subsec:probmodel}, $R_P(0) = N_P$, i.e., it is not possible to distinguish between a positive and a negative sample on the day of the infection. Thus, the number of people infected at the start of a pandemic ($T=0$) are quantified in terms of their deviation from the negative samples on the next day.

Similarly, by repeated application of \eqref{eq:QP}, \eqref{eq:QPestimate}, and \eqref{eq:f0est}, $f(T-1)$ is estimated in terms of the previous time points through the recurrence relation 
\begin{equation}\label{eq:fhatrec}
\widehat{f}(T-1) = \frac{\widehat{Q}_{P}(T) - N_{P}}{R_{P}(1)-N_{P}} - \sum\limits_{t=0}^{T-2} \widehat{f}(t) \frac{R_{P}(T-t)-N_{P}}{R_{P}(1)-N_{P}}.
\end{equation}

We set
\begin{equation}\beta_P(i) = \frac{\widehat{Q}_P(i) - N_P}{R_P(1)-N_P} \quad \text{and} \quad \alpha_P(i) = \frac{R_P(i)-N_P}{R_P(1)-N_P}.\end{equation}
Then, \eqref{eq:fhatrec} becomes
\begin{equation}
\widehat{f}(T-1) = \beta_P(T) - \sum\limits_{t=0}^{T-2} \widehat{f}(t) \alpha_P(T-t).
\end{equation}

Note that this recurrence is defined for all $T \in \{1,2,\cdots\}$. Moreover, provided we use the same domain $D_P$ for all $P$, these equations for all $t$ up to $T-1$ yields the matrix system
\begin{equation}
\begin{bmatrix}
1 & 0 & 0 & \dots & 0 \\
\alpha_P(2) & 1 & 0 & \dots & 0\\
\alpha_P(3) & \alpha_P(2) & 1 & \vdots & 0\\
\vdots & \vdots & \ddots & \vdots & 0 \\
\alpha_P(T) & \alpha_P(T-1) & \dots & \dots & 1
\end{bmatrix} \begin{bmatrix}
\widehat{f}(0) \\ \widehat{f}(1) \\ \vdots \\ \vdots \\ \widehat{f}(T-1)
\end{bmatrix}  = \begin{bmatrix}
\beta_P(1) \\ \beta_P(2) \\ \vdots \\ \vdots \\ \beta_P(T)
\end{bmatrix}.\end{equation}

This matrix system is lower triangular and invertible, and the condition number for the matrix may be high depending upon the choice of $dt$. The system itself can still be solved easily through backward substitution without explicit matrix inversion,
\begin{align} \label{eq:festimate0}
\widehat{f}(0) & = \beta_P(1),\\
\widehat{f}(1) & = \beta_P(2) - \alpha_P(2) \beta_P(1),\\
\widehat{f}(2) & = \beta_P(3) - \alpha_P(3) \beta_P(1) - \alpha_P(2) \beta_P(2),
\end{align}
and so on. The prevalence at time $T$ is then estimated as 
\begin{equation}\label{eq:qestimate}
\widehat{q}(T) = \sum\limits_{t=0}^T \widehat{f}(t).
\end{equation} 
These estimators are unbiased which is proven in the Appendix. Note that this property holds irrespective of the selection of $D_P$. However, the variance of the estimators increases if the probability mass of $D_P$ is too close to $0$ or $1$, similar to the phenomenon reported in Patrone and Kearsley, 2022 \cite{patrone2022minimizing}. Therefore, it is prudent in practice to select $D_P$ carefully to obtain rapid convergence. As a rule of thumb, choosing a rectilinear domain $D_P$ containing $30-70~\%$ of samples is a good starting point. 

\subsection{Optimal Classification Scheme}
\label{subsec:optimal}
We now develop the optimal classification scheme for sample measurements. The aim is to find a sequence of sets $D_P(T)$ and $D_N(T)$ which optimally partition the domain $\Omega$ at time $T$. A measurement $\mathbf{r}$ is classified as positive or negative on day $T$ based on the subset into which this measurement falls.

We utilize our framework to define Borel probability measures (Billingsley, 2008 \cite{billingsley2008probability}) $\mu_P$ and $\mu_N$ which evaluate how much of the probability mass of $P$ and $N$ respectively lies in the set $X \subset \Omega$, 

\begin{equation}
\mu_P(X) = \int_X P(\mathbf{r})\ d\mathbf{r},
\end{equation}

\begin{equation}
\mu_N(X) = \int_X N(\mathbf{r})\ d\mathbf{r}.
\end{equation}

We must ensure that the sequence of sets partition $\Omega$ under $\mu_P$ and $\mu_N$. That is,
\begin{eqnarray}
\mu_P\left( D_P(T) \cup D_N(T) \right) & = \mu_N\left( D_P(T) \cup D_N(T) \right) & = 1,\\
\mu_P\left( D_P(T) \cap D_N(T) \right) & = \mu_N\left( D_P(T) \cap D_N(T) \right) & = 0.
\end{eqnarray}
At a fixed time $T$, if $q(T)$ is known or can be estimated, we can define a loss function as the sum of prevalence weighted average rates of false positives and false negatives as a function of classification domains, 
\begin{equation}\label{eq:ptwiseerr}
\mathcal{L}\left(D_P(T),D_N(T)\right) = \left(1-q(T)\right) \int_{D_P(T)} N(\mathbf{r}) d\mathbf{r} + q(T) \int_{D_N(T)} P(\mathbf{r},T) d\mathbf{r}. 
\end{equation}

A loss function associated with $\Omega_\tau$ is then defined as 

\begin{equation}
\mathcal{L}_\tau \left(\mathbf{D_P}, \mathbf{D_N}\right) = \sum\limits_{T=0}^\tau \mathcal{L}(D_P(T),D_N(T))
\end{equation} 
where $D_P(T)$ and $D_N(T)$ partition the domain $\Omega$ at time $T$, and 
\begin{equation}\mathbf{D_P} = \begin{bmatrix} D_P(0) \\ D_P(1)\\ \vdots \\D_P(\tau)\end{bmatrix}, \quad \mathbf{D_N} = \begin{bmatrix} D_N(0) \\ D_N(1)\\ \vdots \\D_N(\tau)\end{bmatrix}. 
\end{equation}

We use the pointwise loss function \eqref{eq:ptwiseerr} to determine the optimal classification domains $D_P^*(T)$ and $D_N^*(T)$ as an application of Patrone and Kearsley, 2021 \cite{patrone2021classification}. Intuitively, the optimal domain for positive samples is the set where the prevalence weighted probability of the sample being positive with a value $\mathbf{r}$ is larger than that of it being negative. Assuming the boundary set, i.e.\ $\{\mathbf{r}: q(T) P(\mathbf{r}) = \left(1-q(T)\right) N(\mathbf{r})\}$ has measure zero, the optimal domains are

\begin{equation}
D_P^*(T) = \left\{ \mathbf{r} : q(T) P(\mathbf{r},T) > \left(1-q(T)\right) N(\mathbf{r})\right\},
\end{equation}

\begin{equation}
D_N^*(T) = \left\{ \mathbf{r} : \left(1-q(T)\right) N(\mathbf{r}) > q(T) P(\mathbf{r},T) \right\}.
\end{equation}
 
Now, we employ the pointwise optimality at every $T$  to say that for any $\mathbf{D_P}, \mathbf{D_N},$ 
\begin{eqnarray}
\mathcal{L}_\tau \left(\mathbf{D_P}, \mathbf{D_N}\right) & = & \sum\limits_{T=0}^\tau \mathcal{L}(D_P(T),D_N(T))\\
& \geq & \sum\limits_{T=0}^\tau \mathcal{L}(D_P^*(T),D_N^*(T)) = \mathcal{L}_\tau \left(\mathbf{D_P^*}, \mathbf{D_N^*}\right).
\end{eqnarray} 

Thus $\mathbf{D_P^*}, \mathbf{D_N^*}$  defined below are the vectors of optimal classification sets which partition the domain up to $T=\tau$,
\begin{equation}\mathbf{D_P^*} = \begin{bmatrix} D_P^*(0) \\ D_P^*(1)\\ \vdots \\D_P^*(\tau)\end{bmatrix}, \quad \mathbf{D_N^*} = \begin{bmatrix} D_N^*(0) \\ D_N^*(1)\\ \vdots \\D_N^*(\tau)\end{bmatrix}. 
\end{equation}

\section{Results of a Study with COVID Data}
\label{sec:COVID}
As a proof of concept for ideas developed in previous sections, we implement this time-dependent classification scheme on clinical data. We use publicly available dataset associated with Abela et al., 2021 \cite{abela2021multifactorial} which provides antibody measurements for PCR positive individuals along with the days since symptom onset. We use the total SARS-CoV-2 IgG antibody values in MFI-FOE (median fluorescence intensity-fold over empty beads) units as our variable $r$ under consideration. We use one-dimensional $r$ to highlight the effect of time dependence, but our analysis is applicable to data of arbitrary dimensions. See Luke et al., 2022\ \cite{luke2022modeling} and Patrone et al., 2022 \cite{patrone2022optimal}  for additional examples of modeling probability densities for multi-dimensional data. 

The total SARS-CoV-2 IgG data is transformed by using the following logarithmic transform which puts the data on a scale of bits associated with the measurement,
\begin{equation} \label{eq:transform}
Tr(x) = \log_2(x+2)-1.
\end{equation}
Transformed training data for negative individuals is shown as a histogram in Figure \ref{fig:dataneg} and that for positive  individuals plotted against the days since symptom onset is shown in Figure \ref{fig:datapos}.

As the training data does not report corresponding days in the absolute timeline, it is used to only model the probability density functions for antibody response $t$ days after infection. We use gamma distributions to model both the positive response with changing time and the negative distribution (Frank, 2009 \cite{frank2009common}). 

For the negative samples, we use the pre-pandemic measurements from Abela et al., 2021 \cite{abela2021multifactorial} and assume the density function,
\begin{equation} \label{eq:negpdf}
N(r) = \frac{r^{a-1} e^{-r/b}}{\gamma(a) \ b^{a}}.
\end{equation}
Maximum likelihood estimation yields the values of the shape and scale parameters as $a = 17.5825, \ b = 0.1233.$ A histogram for the data and the corresponding probability density function are plotted in Figure \ref{fig:dataneg}.

\begin{figure}
\includegraphics[width = \textwidth]{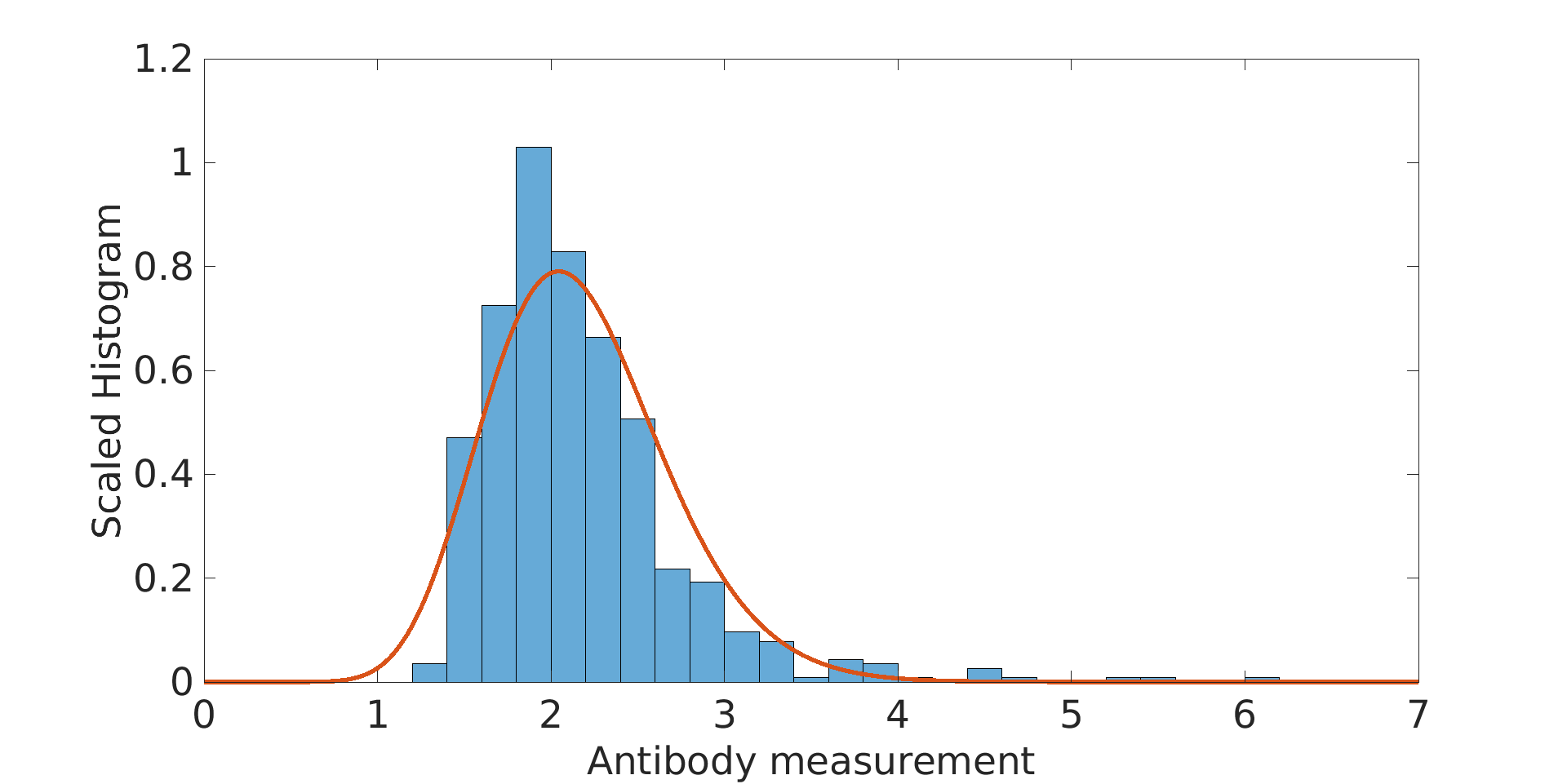}
\caption{The histogram in blue corresponds to transformed training data for the negative samples. The red curve is the probability density function modeled from this data. Refer to equation \eqref{eq:negpdf} in the text for more details.}
\label{fig:dataneg}
\end{figure}

\begin{figure}
\includegraphics[width = \textwidth]{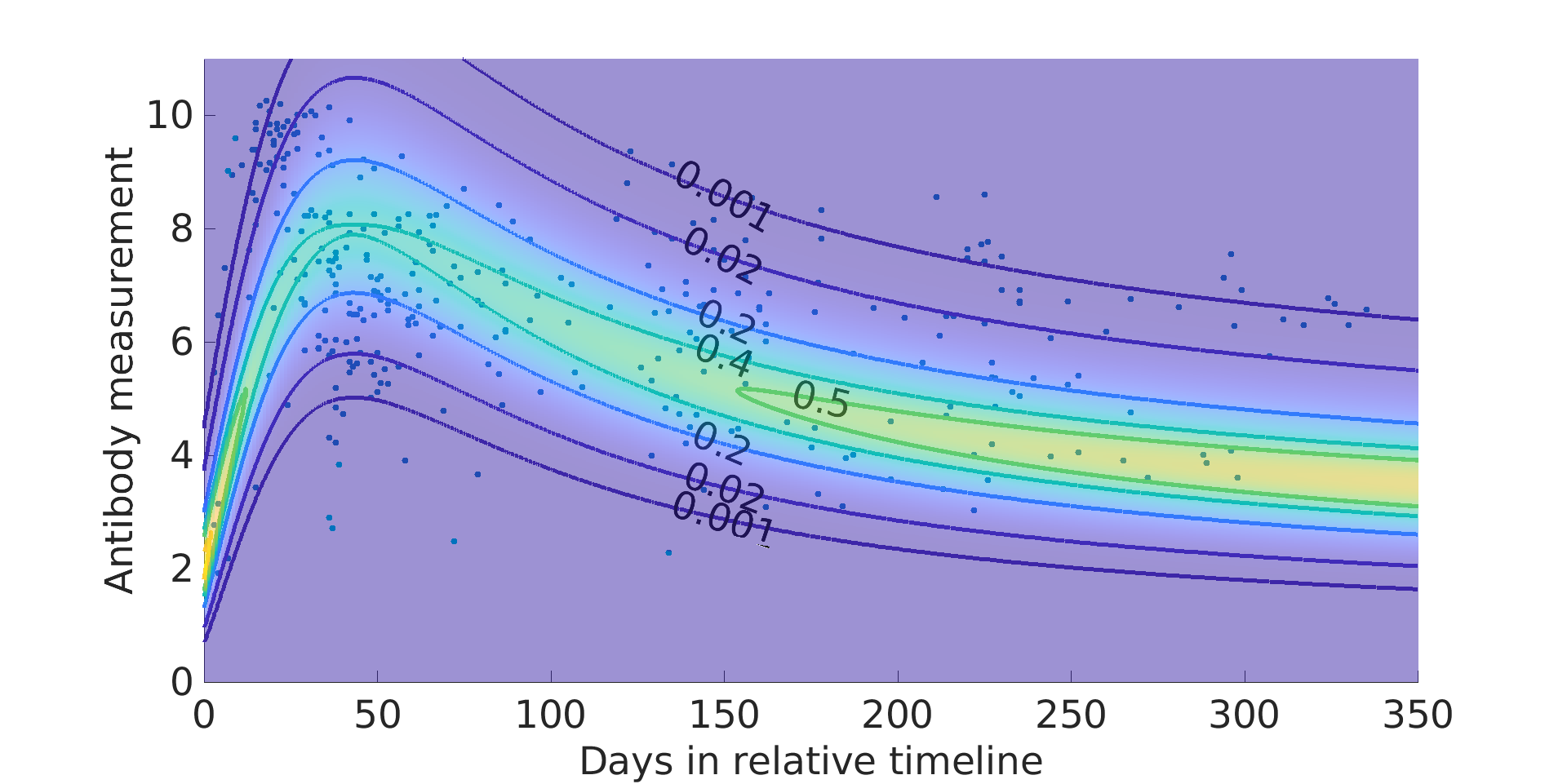}
\caption{Training data for positive samples obtained from Abela et al., 2021 \cite{abela2021multifactorial}. The blue circles represent the positive sample values  after the transformation from equation \eqref{eq:transform} at day t in their personal timeline. The contours show the probability density function modeled this data. Refer to equation \eqref{eq:pospdf} in the main text for more details.}
\label{fig:datapos}
\end{figure}

The antibody data for positive samples is generally scarce for the first few days after infection. To account for this, we use the fact that the antibody levels for those recently infected individual resembles the antibody levels for the uninfected. We supplement the positive samples with twenty negative samples, for which we fix $t=0$.

To take the limiting behavior of antibody kinetics into account, we impose the additional realistic restriction that antibody response for a person infected today ($t=0$) is identical to $N(r)$. We thus model the scale of the gamma distribution for positive samples as a constant independent of time, equal to the scale for the negative samples ($b(t) = b$), whereas the shape of the gamma distribution depends on $t$ in the following way:
\begin{equation}
a(t) = \frac{\theta_1 t}{1+\left(\theta_2 t^2\right)} + a.
\end{equation}
This reflects the known underlying profile in an individual where the antibody levels increase after an infection and then decay slowly over time. The parameters are obtained using maximum likelihood estimates of the training data,
\begin{equation} \label{eq:pospdf}
R(r,t) = \frac{r^{a(t)-1} e^{-r/b}}{\gamma(a(t)) \ b^{a(t)}},
\end{equation}
with \begin{equation}\theta_1 = 2.2251, \ \theta_2 = 0.0005.\end{equation}
The scatterplot for the data and the contours for probability density function obtained is plotted in Figure \ref{fig:datapos}.

\subsection{Prevalence Estimation with Synthetic Data}

We next demonstrate the behavior of the estimators for $f$ and $q$, using synthetic data generated from probability models in  \eqref{eq:negpdf}, \eqref{eq:pospdf}, along with an assumed prevalence. The prevalence is then estimated using the scheme developed in subsection \ref{subsec:qest} and compared with the true values.

Discretization for time is chosen as $dt = 14$ so that $10$ such time periods are $140$ days in the absolute timeline. A $1000$ sets of synthetic data are then generated for underlying known $f$ (constant and sinusoidal) over this time. The mean and variance of the prevalence estimates over these multiple sets of synthetic data with $N_s$ sample points in each time interval are discussed below.

\begin{figure}
\includegraphics[width = \textwidth]{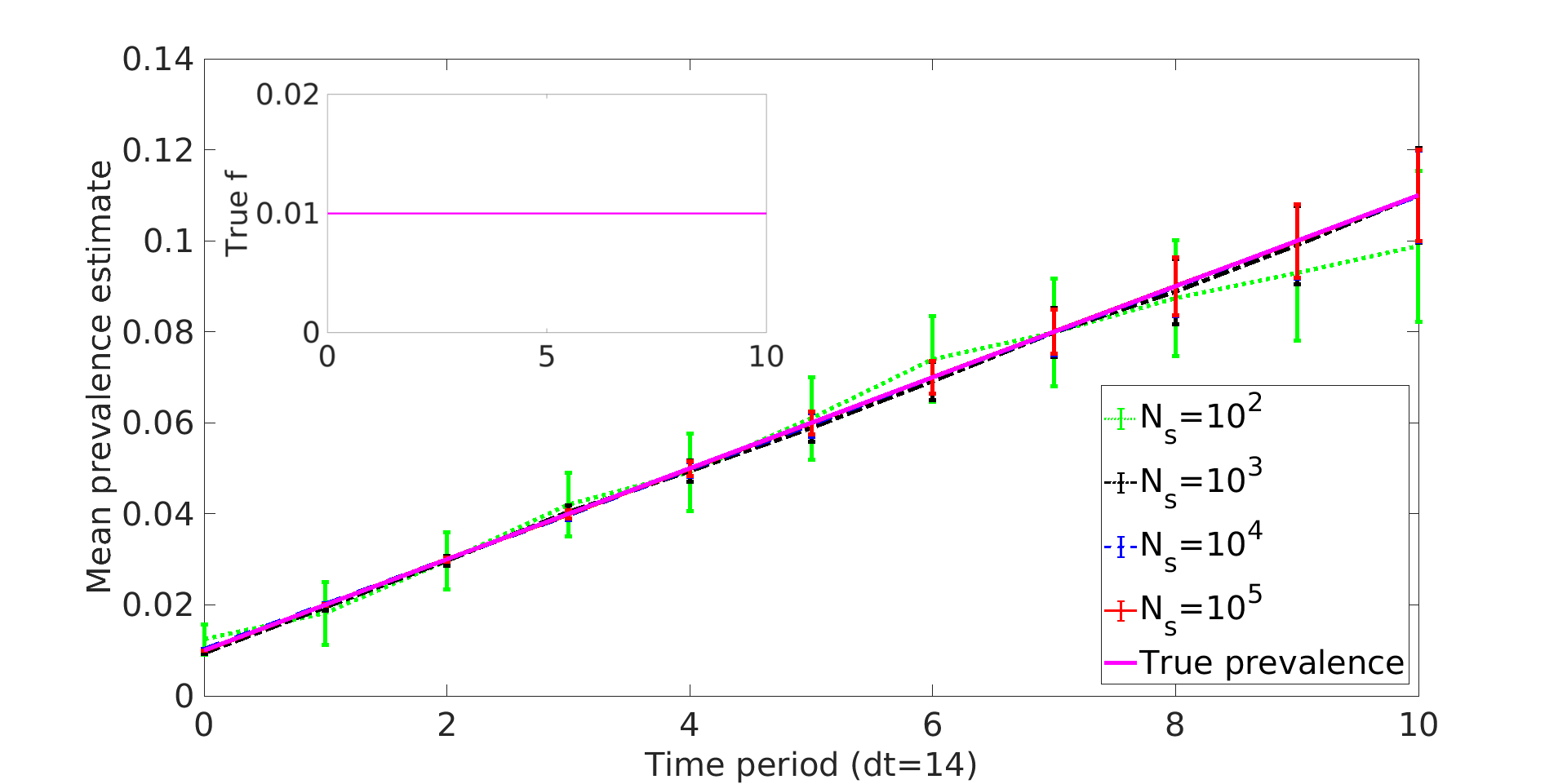}
\caption{Mean over 1000 synthetic sets of prevalence estimates. The inset shows the underlying true constant $f$, defined by \eqref{eq:fconst}. The mean prevalence estimates for various $N_s$ number of samples per time period are plotted with variance errorbars over time: $N_s=100$ (Green, dotted), $N_s = 1000$ (Black, dotted-dashed), $N_s =10^4$ (Blue, dashed), $N_s = 10^5$ (Red, solid), True prevalence (Magenta).}
\label{fig:constant_mean}
\end{figure}

\begin{figure}
\includegraphics[width = \textwidth]{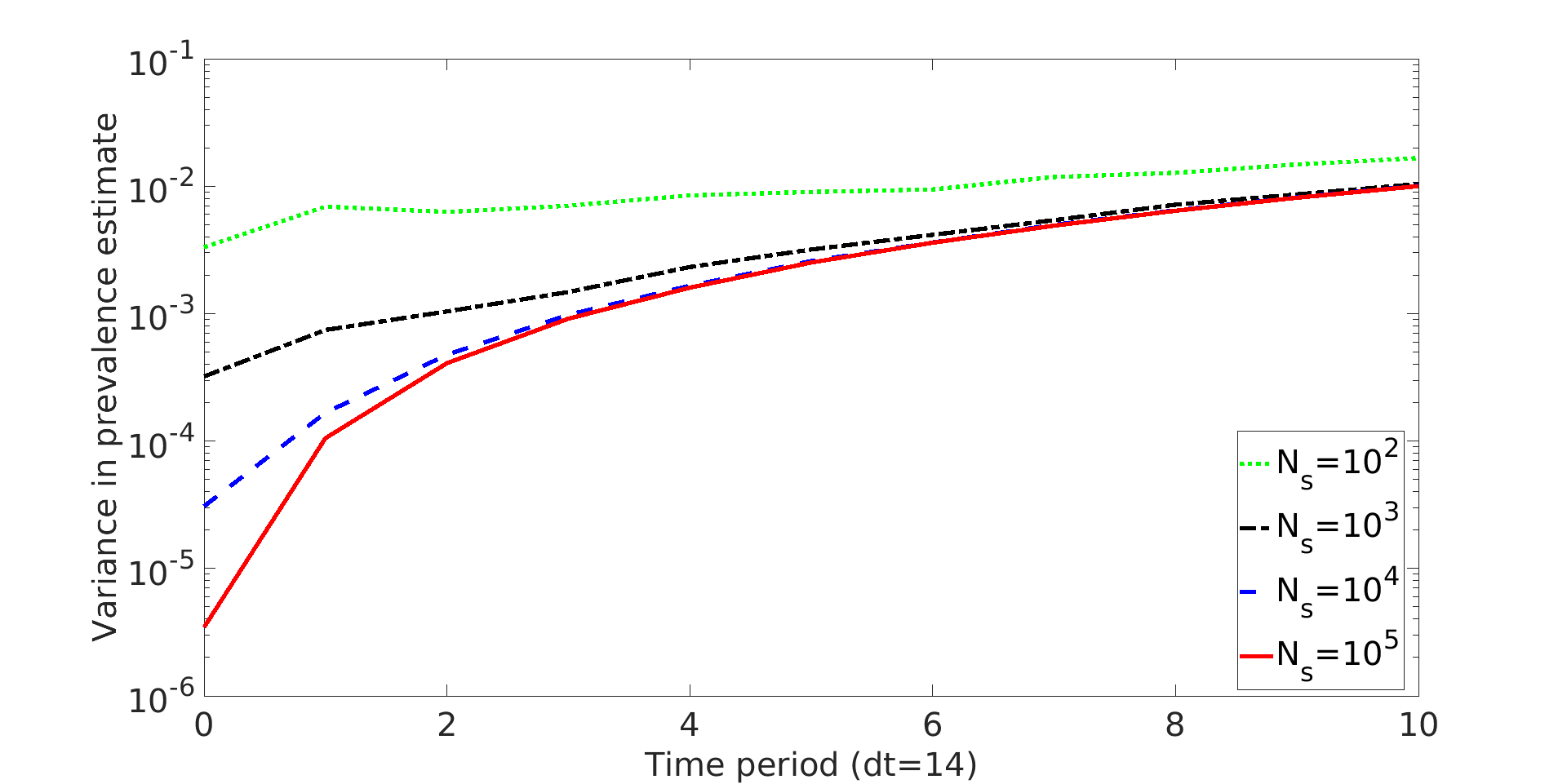}
\caption{Variance over 1000 synthetic sets of prevalence estimates, underlying true constant $f$, defined by \eqref{eq:fconst}. The variance of prevalence estimates for various $N_s$ number of samples per time period are plotted: $N_s=100$ (Green, dotted), $N_s = 1000$ (Black, dotted-dashed), $N_s =10^4$ (Blue, dashed), $N_s = 10^5$ (Red, solid).}
\label{fig:constant_var}
\end{figure}

For a constant change in prevalence per time period, i.e.\ \begin{equation} \label{eq:fconst}
f(i) = 0.01, \quad i \in \{0,1,\cdots,10\},
\end{equation}
Figure \ref{fig:constant_mean} shows the  mean of prevalence estimates over time for various values of $N_s$. The errorbars show the associated variances. As expected, the mean prevalence estimates matches the true underlying prevalence for large enough $N_s$. Due to accumulation of errors over time, the variances for a given $N_s$ increase with time as indicated by the size of the errorbars and the graph of the variances in Figure \ref{fig:constant_var}.

Consider a sinusoidal change in prevalence per time period, i.e.\ \begin{equation} \label{eq:fwave}
f(0) = 0.01, \quad f(i) = 0.01 \sin\left( \frac{(i-1) \pi}{10}\right), \quad i \in \{1,\cdots,10\}.
\end{equation}
This $f$ emulates a wave of infections in a pandemic. Figure \ref{fig:wave_mean}  displays the unbiasedness of the estimates by using synthetic data from this wave. The variance for the estimates unsurprisingly increases with time (Figure \ref{fig:wave_var}); the errors in estimation accumulate over time. Even so, using a larger number of samples per time period helps decrease this the effect to a certain degree.

\begin{figure}
\includegraphics[width = \textwidth]{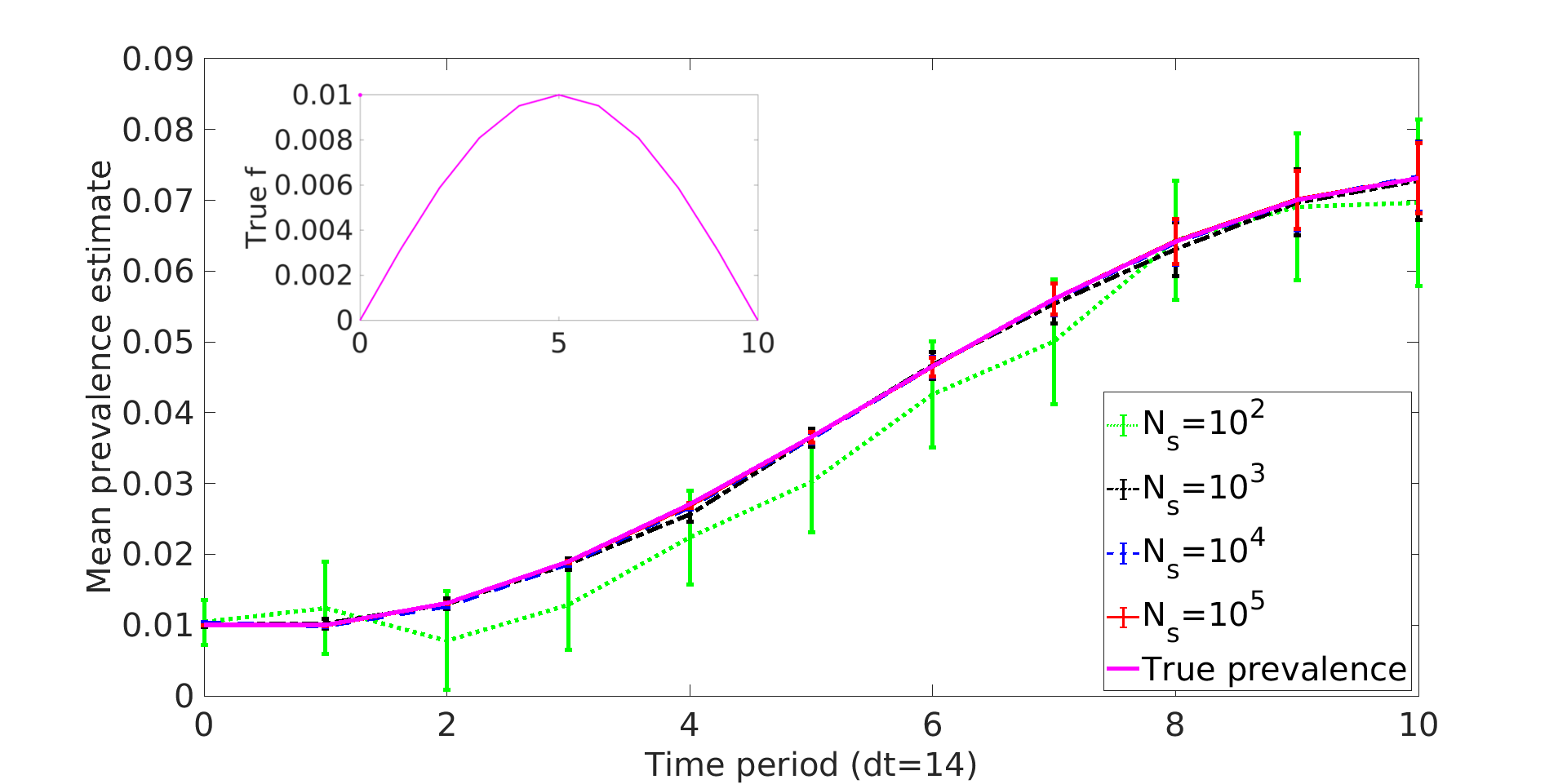}
\caption{Mean over 1000 synthetic sets of prevalence estimates. The inset shows the underlying true sinusoidal $f$, defined by \eqref{eq:fwave}. The mean prevalence estimates for various $N_s$ number of samples per time period are plotted with variance errorbars over time: $N_s=100$ (Green, dotted), $N_s = 1000$ (Black, dotted-dashed), $N_s =10^4$ (Blue, dashed), $N_s = 10^5$ (Red, solid), True prevalence (Magenta).}
\label{fig:wave_mean}
\end{figure}

\begin{figure}
\includegraphics[width = \textwidth]{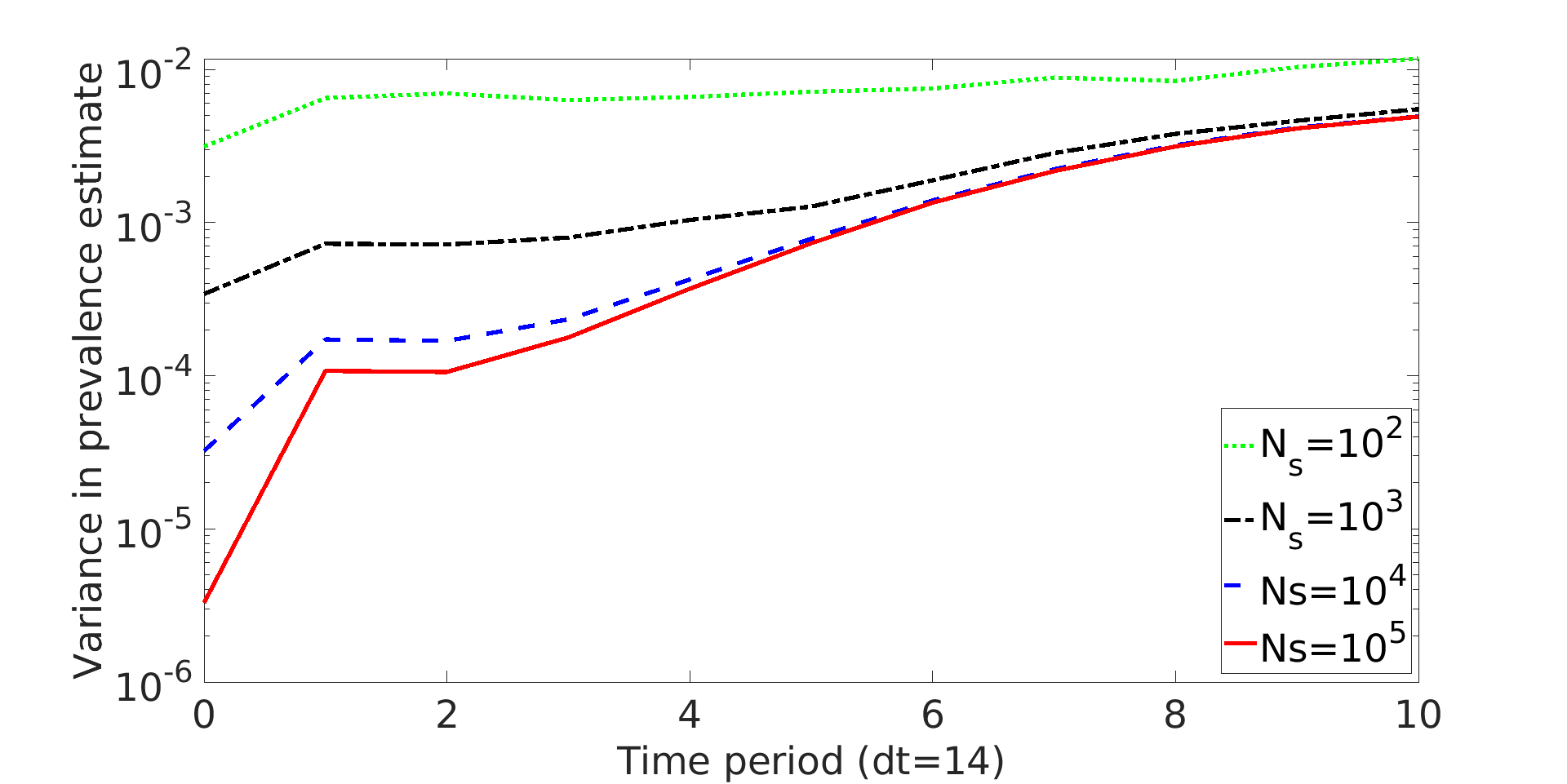}
\caption{Variance over 1000 synthetic sets of prevalence estimates, underlying true sinusoidal $f$, defined by \eqref{eq:fwave}. The variance of prevalence estimates for various $N_s$ number of samples per time period are plotted: $N_s=100$ (Green, dotted), $N_s = 1000$ (Black, dotted-dashed), $N_s =10^4$ (Blue, dashed), $N_s = 10^5$ (Red, solid).}
\label{fig:wave_var}
\end{figure}

\subsection{Optimal Classification Domains with Synthetic Data} \label{subsec:optsynth}Once the prevalence has been estimated, we calculate the optimal classification domains using the scheme developed in subsection \ref{subsec:optimal}. We investigate the effect of different underlying pandemic prevalences on these time-dependent classification domains using different known $f$ to demonstrate the behavior of the classification boundary over time.

In Figure \ref{fig:scatterimpulse}, we plot the classification domains for the case in which $f(T)$ is an impulse, i.e., 
\begin{equation} \label{eq:fimpulse}
f(0) = 0.01, \quad f(t) = 0, t>0.
\end{equation}
This figure illustrates optimal classification domains if a fraction of the population is infected on day $0$, but no additional infections occur. The optimal classification boundary still changes with time due to varying antibody levels associated with the personal timeline of infection. For a fixed day in the absolute timeline, the boundary between the positive and the negative domains is the antibody measurement threshold for classification on that day. \emph{Notice that this boundary varies in time, and as a result, the same antibody measurements on different days can be classified differently.} We explore the justification and interpretation of this in subsection \ref{subsec:interpret}.

The steep change in the optimal classification boundary at $t=0$ also reconciles the extreme example considered earlier. Before a pandemic starts, the boundary between positive and negative samples is set at a high value, as every sample is classified as negative. This boundary quickly falls when we start labeling samples as positive with changing time.  

The plot highlights quantitative importance of time dependence. Notice that the difference between the values with the maximum and the minimum between day $10$ and $300$ are as high as $1.25$. When translated to original measurement data, the relative change is even higher. The standard classification scheme that does not take time dependence into account might thus potentially misclassify a large number of samples.

The location of the optimal classification boundary in the case of constant prevalence as in \eqref{eq:fimpulse} changes solely due to the changing probability density functions with time. At very small times, the boundary is still at a large antibody value, as there is no way to meaningfully separate the positive samples from the negative samples as seen in Figure \ref{fig:posnegcomb}. However, as time progresses, the optimal boundary is determined by the intersection of prevalence weighted positive and negative probability densities. The shape parameter for the gamma distribution of our model for the positive response also determines the location of the optimal boundary.

\begin{figure}
\includegraphics[width = \textwidth]{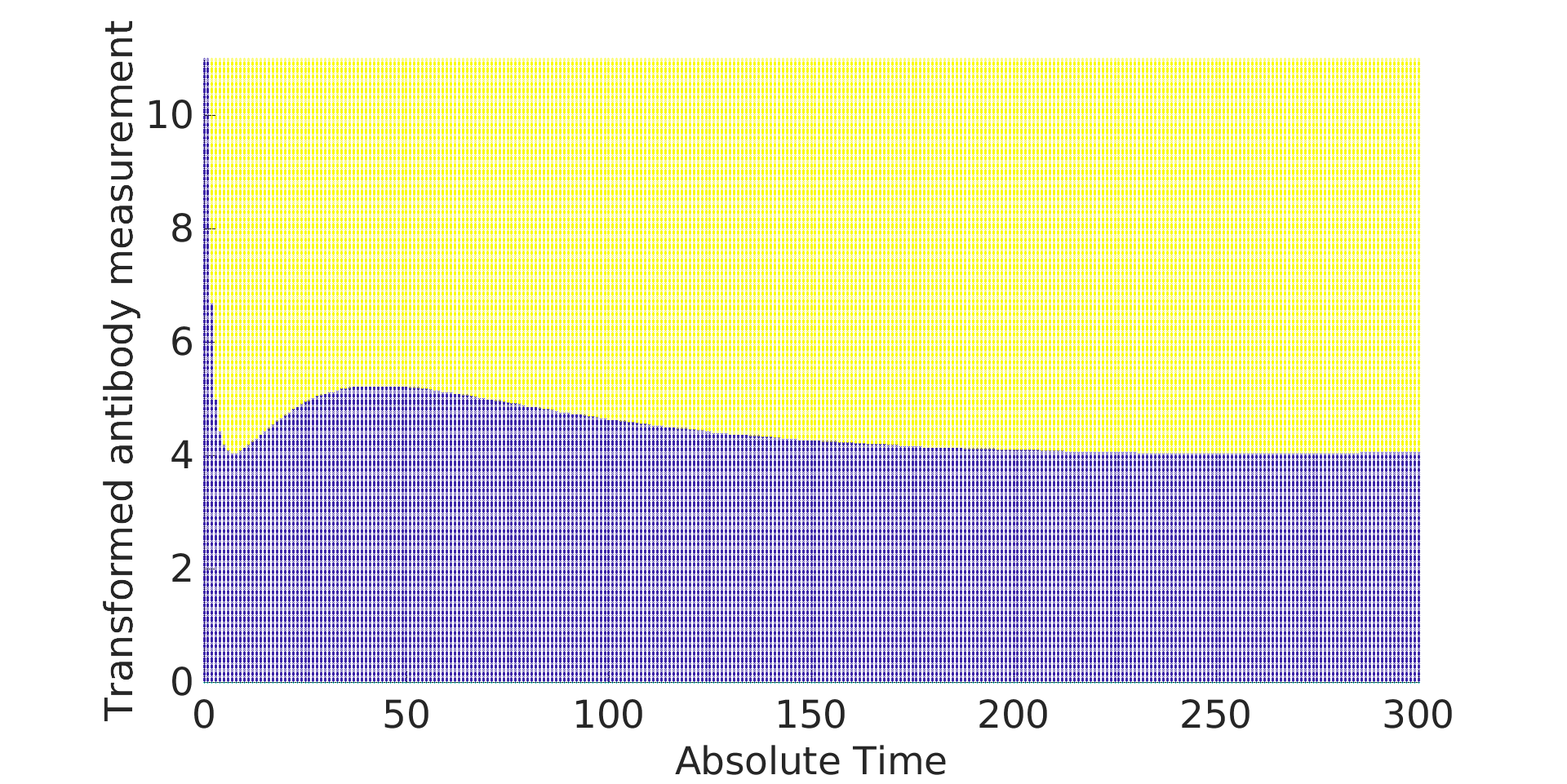}
\caption{Optimal classification partition of the domain $\Omega\times\{0,1,\cdots,300\}$ under an impulse given in \eqref{eq:fimpulse}. The yellow represents positively classified values and the blue represents negatively classified values.}
\label{fig:scatterimpulse}
\end{figure}

\begin{figure}
\includegraphics[width = \textwidth]{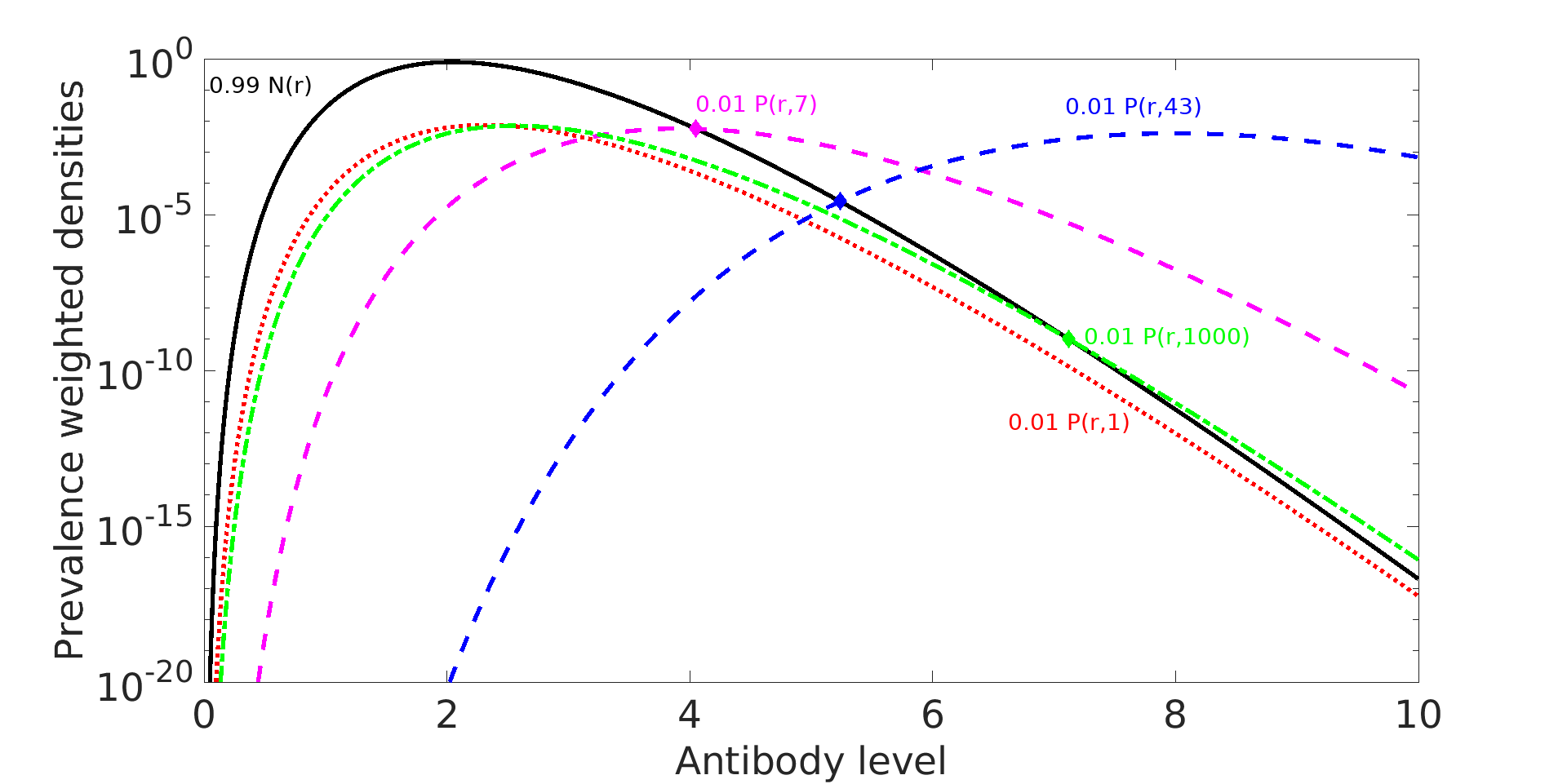}
\caption{Optimal classification partition of the domain $\Omega$ under an impulse given in \eqref{eq:fimpulse}. The black solid line is the prevalence weighted density for the negative samples. The prevalence weighted density for the positive samples at $T=1$ (red, dotted), $T=7$ (pink, dashed), $T=43$ (blue, dashed), $T=1000$ (green, dotted-dashed). $T=1$ and $T=1000$ demonstrate the behavior at the beginning and asymptotically. $T=7$ and $T=43$ are chosen as they are the local minima and local maxima of the boundary. }
\label{fig:posnegcomb}
\end{figure}

\begin{figure}
\includegraphics[width = \textwidth]{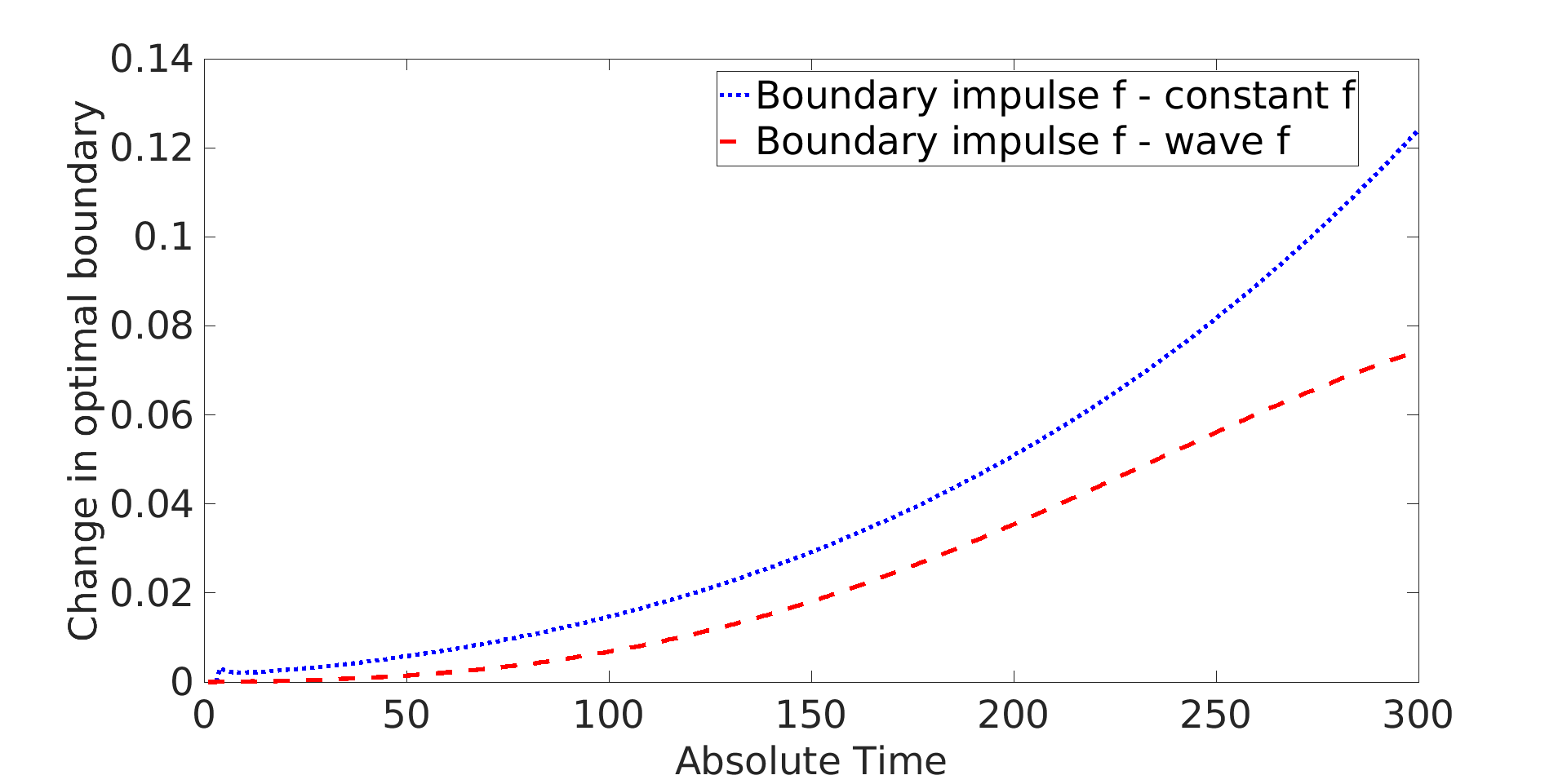}
\caption{Difference in optimal classification boundaries for impulse $f$ from that of constant $f$ (blue, dotted) and from that of wave $f$ (red, dashed). All $f$ have the same scale, $0.001$.}
\label{fig:diffffimp}
\end{figure}

\begin{figure}
\includegraphics[width = \textwidth]{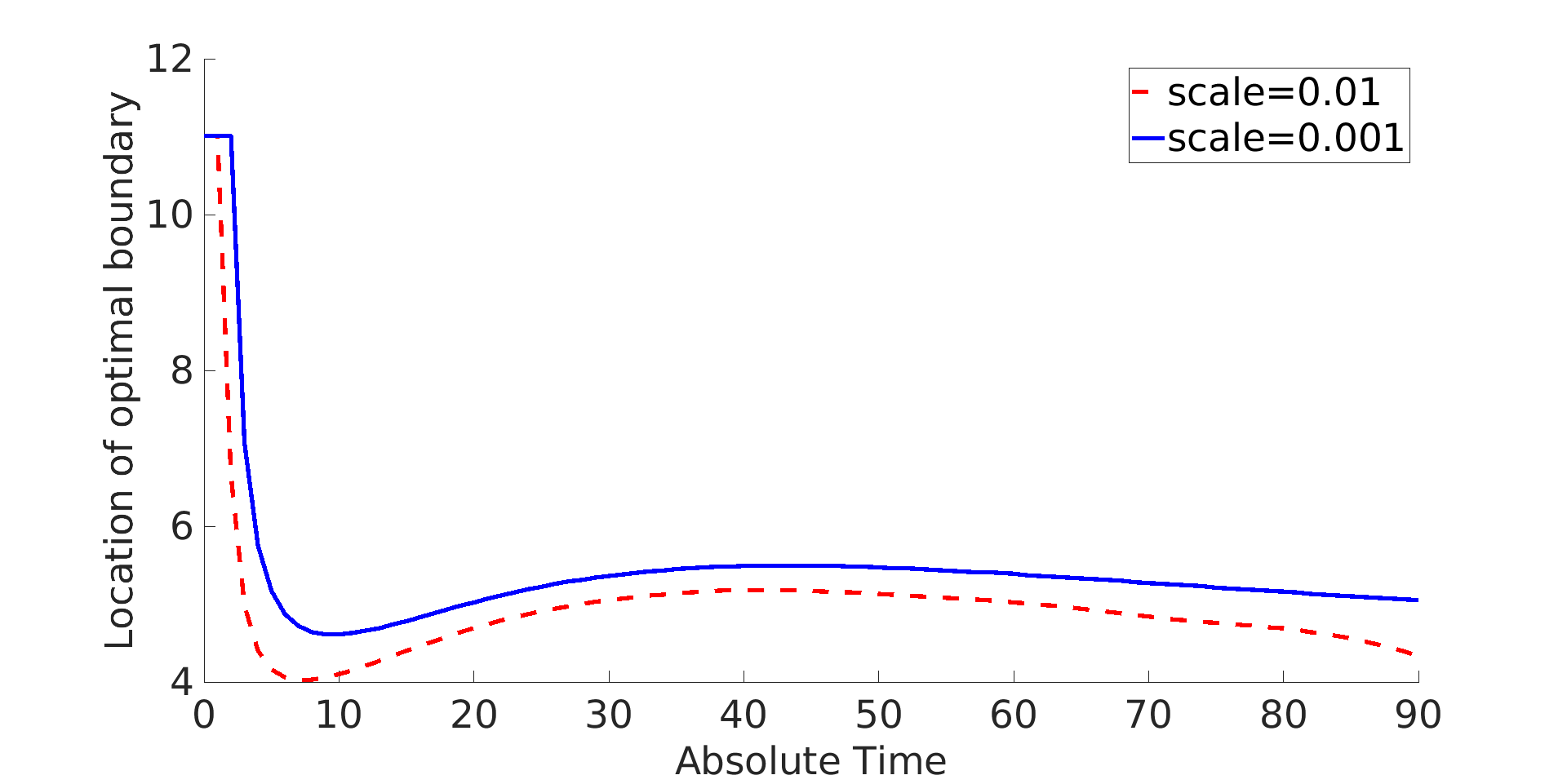}
\caption{Location of optimal classification boundaries for same $f$ scaled to different magnitudes. Scale $0.01$ (red, dashed) and 0.001 (blue, solid).}
\label{fig:diffscale}
\end{figure}

\begin{figure}
\includegraphics[width = \textwidth]{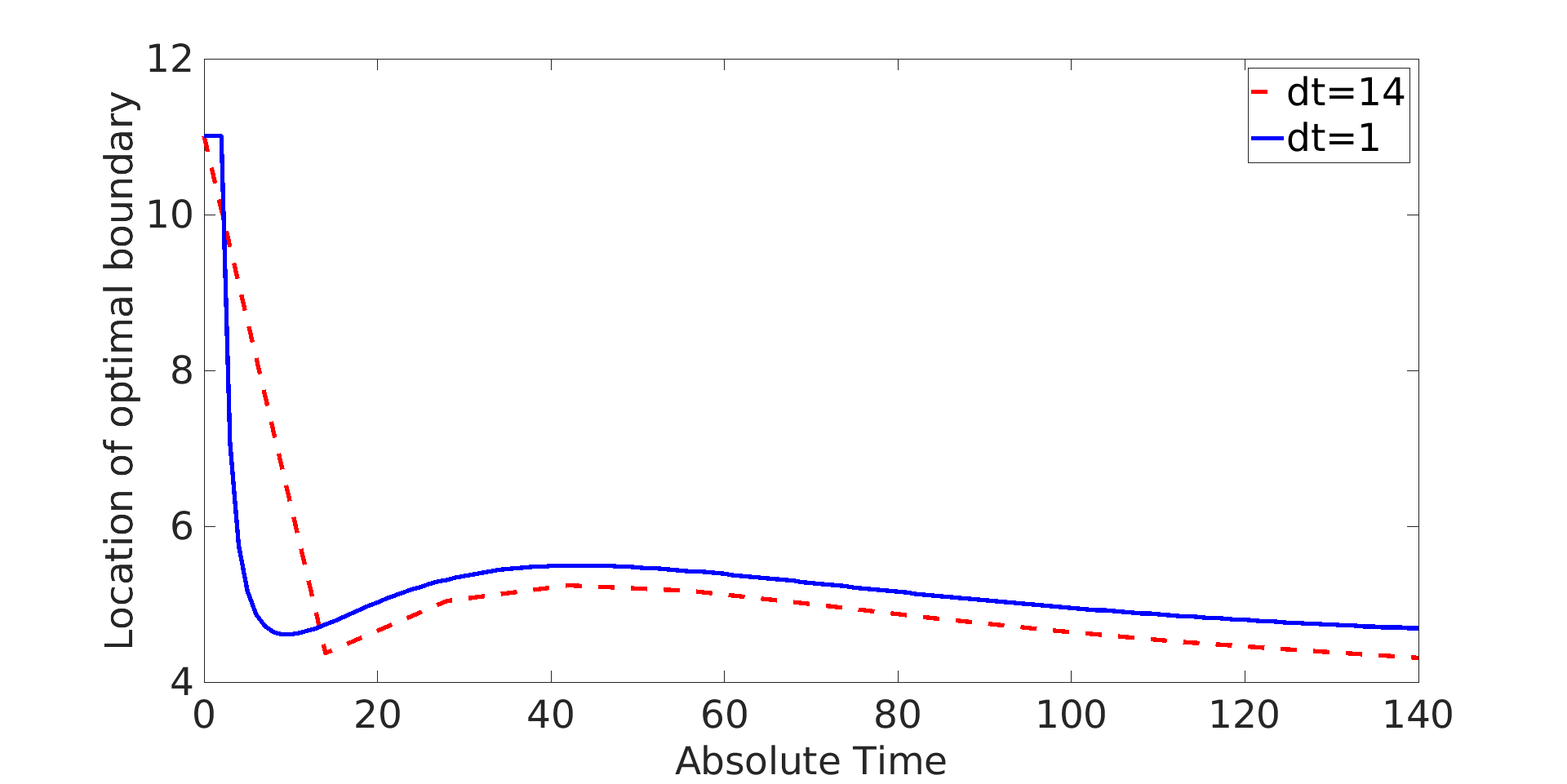}
\caption{Location of optimal classification boundary with different time discretizations. $dt = 14$ (red, dashed) and $dt = 1$ (blue, solid).}
\label{fig:diffdt}
\end{figure}

Even for different rates of infection $f$, the shape of the classification boundary can be qualitatively similar, i.e.\ the shape of the optimal boundary over time depends largely on $R(\mathbf{r},t)$ and $N(\mathbf{r})$ as demonstrated in Figure \ref{fig:diffffimp}. However, different magnitude scales of $f$ for the same type of function have boundaries that are different from each other as demonstrated in Figure \ref{fig:diffscale}. The plot shows the classification boundaries for two such $f$. Notice that the higher the prevalence at a given time, the lower the classification boundary at that time. To gain intuition, one can observe that for the extreme case of zero prevalence for all times (the absence of a pandemic), the classification boundary stays at the largest possible $r$. This is because the optimal strategy is to classify every sample as negative irrespective of the measurement value.

Figure \ref{fig:diffdt} highlights the effects of picking different time discretizations. For a $140$ day period in the absolute timeline, with
\begin{equation}
f(i) = 0.01, \quad i\in \{0,1,2,\cdots,140\},
\end{equation}
we show the optimal domain boundaries for $dt=1$ and $dt=14$. Notice that $dt=14$ provides a much coarser and less desirable approximation than $dt=1$. However, when data is sparse, taking $dt$ to be larger helps obtain a superior prevalence estimate. In practice, therefore, an attempt should be made to balance these opposing effects by gathering as much data as possible.
\section{Discussion and Conclusions}

\subsection{Continuous Time Extension}
\label{subsec:cts}
The time-discretized model described in subsections \ref{subsec:probmodel} and \ref{subsec:qest} can be viewed as the discretization of a continuous model with $\tilde{f}$ defined on $(-\infty,T]$. The corresponding values of $f(n)$ for the discrete case are
\begin{equation}f_0 = \int_{-\infty}^0 \tilde{f}(t) dt, \quad f(n) = \int_{n-1}^n \tilde{f}(t) dt,\end{equation}
with the prevalence at a given day $N$ being
\begin{equation}q(N) = \sum\limits_{n=0}^N f(n) = \int_{-\infty}^N \tilde{f}(t) dt.\end{equation}
Even in the continuous domain, $\tilde{f}$ is the change in prevalence. Similar to \eqref{eq:PrT}, conditional probability of positive samples is 
\begin{equation}\tilde{P}(\mathbf{r},T) = \int_{-\infty}^T \tilde{R}(\mathbf{r},t) \frac{\tilde{f}(T-t)}{\tilde{q}(T)} dt.\end{equation}
Analogous to the discrete case, one finds 
\begin{align}\tilde{Q}_P(T) & = N_P + \int_{-\infty}^T \left(R_P(T-t)-N_P\right) \tilde{f}(t) dt\\ & = N_P + (\tilde{f}\mathds{1}_{(-\infty,T]} * g) (T).\end{align}

Here, $*$ denotes continuous convolution, and $g(t) = R_P(t)-N_P$, meaning 
\begin{equation} (\tilde{f}\mathds{1}_{(-\infty,T]} * g) (T) = \tilde{Q}_P(T) - N_P.\end{equation}
Deconvolution provides a natural scheme for estimation of $\tilde{f}$. Care needs to be taken to explicitly define function spaces for $\tilde{f}$. Interestingly, this estimate cannot be obtained as a straightforward limit of \eqref{eq:fhatrec} as $\Delta t \rightarrow 0$ because in this limit, the denominator tends to zero.

\subsection{Interpretation} \label{subsec:interpret}
The essence of the work is that the time-dependent nature of antibody levels in an infected individual as well as the progression of a pandemic both change the day-to-day probabilities of a measurement being positive. As a result, the threshold value for whether a sample is classified as positive changes with time.

This observation implies that time $T$ and measurement $\mathbf{r}$ are both variables for classification, that is, time itself is elevated to the same standing as measurements $\mathbf{r}$. As we are not surprised that two different $\mathbf{r}$ values on the same day are classified differently, we need not be surprised that a particular $\mathbf{r}$ value on two different days $T$ can be classified differently. An extreme example mentioned earlier shows that the notion of time dependence is implicitly widely accepted; all samples are classified negative before pandemic begins irrespective  of the antibody measurement, which is no longer the case afterwards.

\subsection{Modeling Choices}
\label{subsec:modelingchoices}
The use of probability models allows one to leverage pre-existing knowledge of kinetic antibody profiles for infected individuals (Ortega et al., 2021\ \cite{ortega2021seven}, Qu et al., 2020 \cite{qu2020profile}). As a result, the probability density functions can be modeled beforehand, and the parameters for the same can be adjusted as a pandemic progresses and new training data is collected. This leads to a classification scheme with lower error rate (as defined in subsection \ref{subsec:optimal}) than a method based on $3\sigma$ confidence interval updated sporadically, or other non-probabilistic machine learning models. Explicit probability modeling also circumvents the assumption of Gaussian distribution inherent to the $3\sigma$ scheme.

This proposed method is assay and infection agnostic, in the sense that as long as the same assay is used to develop the response probability densities as the data to be classified, the analysis should be valid. Modeling for special circumstances could be an interesting extension. For instance, data concentrated at the limits of detection can be modeled with mixed probability densities involving Dirac delta distribution (Patrone et al., 2022 \cite{patrone2022optimal}). 

Our analysis can be used even in low prevalence settings, where traditional methods have trouble. The probability densities pertaining to antibody response ($R$) can be modeled given sufficient training data, even if the local prevalence is low. Notice that the estimate of $Q_P(T)$ is better with more data points, and that the estimate does not depend on the underlying prevalence. Limited testing capability can however hamper the prevalence estimation. In such a scenario, a sensible choice of $D_P$ helps achieve the prevalence estimate more quickly as outlined in Patrone and Kearsley, 2022 \cite{patrone2022minimizing}.

\subsection{Limitations and Future Work}\label{subsec:limitations}
The effects of vaccine-induced immunity and reinfection could be studied in detail and explicitly considered in the model. Different variants of a disease can generate different antibody response profiles. Moreover, protection provided by natural infection and immunization can decay over time. These considerations are beyond the scope of this manuscript and will be taken into account in future work once more data becomes available.

We did not use the days since infection as determined by PCR or the days since exposure to set relative time for a few reasons. First, it is extremely difficult to obtain data that accurately captures this information. Moreover, antibodies are formed after an initial delay when the immune system mounts its response against the virus. This delay depends on the class of antibody measured (IgG, IgM, IgA), the virus variant under consideration, the vaccination status, among other factors (Sethuraman et al., 2020\ \cite{sethuraman2020interpreting}, Muecksch et al., 2022 \cite{muecksch2022longitudinal}, Zhong et al., 2021 \cite{zhong2021durability}), and would be difficult to model effectively without robust longitudinal studies. However, days since symptom onset is a highly subjective quantity, and its use is not ideal. More thorough studies can help reduce the modeling error introduced with this choice.

Antibody production in an individual varies with a multitude of factors like age, sex, other diseases to name a few. These complex effects can be considered explicitly to expand this model in the future.

\subsection{Recommendations}
Data including the measurements and true classes for a large number of samples along with the true date of infection are not abundant.  As a remedy, we recommend undertaking longitudinal studies from the beginning of any pandemic to extract as much information as possible. Moreover, considering the low costs associated with saving numeric datasets, we recommend public health authorities to preserve such de-identified datasets with the highest granularity possible. Such training datasets can help in real time application of optimal estimation and classification schemes.

\section*{Acknowledgments}

This work is a contribution of the National Institute of Standards and Technology and is not subject to copyright in the United States. P.B. was funded through the NIST PREP grant 70NANB18H162. The aforementioned funder had no role in study design, data analysis, decision to publish, or preparation of the manuscript. The authors wish to thank Matthew DiSalvo and Rayanne Luke for useful discussions during preparation of this manuscript.

\section*{Research involving Human Participants
and/or Animals}

Use of data provided by Abela et al., 2021 \cite{abela2021multifactorial} has been reviewed and approved by the NIST Research Protections Office.

\section*{Data Availability}

Analysis scripts and data developed as a part of this work are available upon reasonable request. Original data are provided in Abela et al., 2021 \cite{abela2021multifactorial}.

\section*{Declarations of Competing Interests} None.

\appendix

\section{Unbiasedness of the Estimators}

\begin{lemma} \label{lem:QP}
For all $T \in \mathbb{N}, \widehat{Q}_P(T)$ is  an unbiased estimator of $Q_P(T)$. 
\end{lemma}
\begin{proof}
The indicator for each sample value $\mathbf{r}_i$ in \eqref{eq:QPestimate} is distributed as Bernoulli with the mean $Q_P(T)$. Adding $S$ such i.i.d. distributions and using  \eqref{eq:QPestimate}, we say that 
\begin{equation*}
\widehat{Q}_P(T) \sim \frac{1}{S} \text{Binomial}(S,Q_P(T)),
\end{equation*}
as the probability of a randomly selected sample on day $T$ to be in the set $D_P$ is $Q_P(T)$. Therefore,
\begin{equation*}
E(\widehat{Q}_P(T)) =  \frac{1}{S} E\left(\text{Binomial}(S,Q_P(T))\right) = Q_P(T).
\end{equation*}

\end{proof}
\begin{theorem}
The estimators in equations \eqref{eq:festimate0}-\eqref{eq:qestimate} are unbiased.
\end{theorem}

\begin{proof}
Taking expectation of both sides for \eqref{eq:festimate0},
\begin{equation*}E\left(\widehat{f}(0)\right) = E\left(\frac{\widehat{Q}_P(1)-N_P}{R_P(1)-N_P}\right) =\frac{E\left(\widehat{Q}_P(1)\right)-N_P}{R_P(1)-N_P}.
\end{equation*}

We use Lemma \ref{lem:QP} to say,
\begin{equation*}E(\widehat{f}(0)) = f(0).
\end{equation*}

Using induction,
\begin{align*}
E\left(\widehat{f}(T-1)\right) & = \frac{E\left(\widehat{Q}_{P}(T)\right) - N_{P}}{R_{P}(1)-N_{P}} - \sum\limits_{t=0}^{T-2} E\left(\widehat{f}(t)\right) \frac{R_{P}(T-t)-N_{P}}{R_{P}(1)-N_{P}}\\
& = \frac{Q_{P}(T) - N_{P}}{R_{P}(1)-N_{P}} - \sum\limits_{t=0}^{T-2} f(t) \frac{R_{P}(T-t)-N_{P}}{R_{P}(1)-N_{P}}\\
& = f(T-1)
\end{align*}
and 
\[E\left(\widehat{q}(T)\right) = E\left(\sum\limits_{t=0}^T \widehat{f}(t) \right) = q(T).\]
\end{proof}

\bibliographystyle{plain} 
\bibliography{refs} 
\end{document}